%% file: main.tex
\documentclass[aps,prx,superscriptaddress,twocolumn]{revtex4-2}

\usepackage{mathtools}
\usepackage{amssymb}
\usepackage[dvipsnames,x11names,table]{xcolor}
\usepackage{tcolorbox}
\usepackage{tikzit}
\usepackage{graphicx}
\usepackage{caption}
\usepackage{amsthm}
\usepackage{thmtools}
\usepackage{hyperref}
\usetikzlibrary{decorations.pathreplacing}
\input{tikzinfo}

\hypersetup{colorlinks=true, linkcolor=BrickRed, citecolor=BrickRed, urlcolor=BrickRed}

\newcommand{\ket}[1]{\ensuremath{|#1\rangle}}
\newcommand{\bra}[1]{\ensuremath{\langle #1|}}

\newcommand{\ketbra}[2]{\ensuremath{|#1\rangle\!\langle #2|}}
\newcommand{\dk}[1]{\ensuremath{|#1\rangle\!\rangle}}
\newcommand{\db}[1]{\ensuremath{\langle\!\langle #1|}}
\newcommand{\Ptil}{\ensuremath{\widetilde{\mathcal P}}}
\newcommand{\Ctil}{\ensuremath{\widetilde{\mathcal C}}}

\DeclareMathOperator{\tr}{\text{tr}}
\DeclareMathOperator{\var}{Var}

\definecolor{boxborder}{RGB}{148,209,190}

\newtcolorbox[auto counter]{resultbox}[2][]{%
    colback=white, colframe=boxborder, coltitle=black,
    before skip=\baselineskip, after skip=\baselineskip,
    title={Result~\thetcbcounter: #2}, label={#1}}

\newtheorem{theorem}{Theorem}

\newcommand{\Var}{\var}
\newcommand{\vv}[1]{\mathbf{#1}}

\begin{document}

\title{Stacking the Deck:\\Tunable Trainability in Stacked LCUs}

\author{Nikhil Khatri}
\author{Stefan Zohren}
\affiliation{Machine Learning Research Group, University of Oxford, United Kingdom}
\author{Gabriel Matos}
\affiliation{Quantinuum, London, United Kingdom}

\date{27 July 2026}

\begin{abstract}
Variational quantum circuits have been central to many proposed near-term applications of quantum computing, but a growing body of evidence suggests that trainability and quantum advantage are fundamentally at odds: ans{\"a}tze expressive enough to resist efficient classical simulation tend to exhibit barren plateaus, while structures that provably rule out barren plateaus typically render them classically simulable. We propose a stacked linear combination of unitaries (S-LCU) as a variational ansatz which provides a tunable trade-off between barren plateaus and classical simulability.
Using a diagrammatic analysis, we bound the loss-landscape variance of the \emph{Free Fermion S-LCU}, whose elements are fermionic Gaussian unitaries.
We prove a variance lower bound of $\Omega(1/(n k^{3l}))$, with a simulation cost of $O(k^{2l} n^3)$ using the best known classical algorithm, compared to a quantum gate complexity of only $O(lkn^2)$. The number of layers $l$ serves as a single dial that trades computational complexity against the rate of cost concentration. This offers practitioners a systematic method for constructing ans{\"a}tze with a complexity-trainability trade-off that best suits their application and hardware.
\end{abstract}

\maketitle

\section{Introduction}
Variational quantum algorithms optimise a parametrised quantum circuit against a cost function, and are an important tool in many proposed near-term applications of quantum computing, from quantum machine learning~\cite{benedetti2019parameterized,biamonte2017qml} to ground state preparation~\cite{peruzzo2014variational,mcclean2016theory,tilly2022vqe} and combinatorial optimisation~\cite{farhi2014qaoa,blekos2024qaoa,abbas2024optimization}. Their utility as quantum routines relies on a balance between two demands: a parameterised circuit must be efficiently trainable, yet be difficult to simulate classically. These goals are increasingly understood to be in tension. Models with sufficiently high expressivity tend to exhibit \emph{barren plateaus}, where the variance of the gradient vanishes exponentially with system size, rendering training from a random parameter initialisation infeasible~\cite{mcclean2018barren,cerezo2021cost,Larocca2025}. Conversely, the structure that provably rules out barren plateaus typically also enables efficient classical simulation~\cite{cerezo2023simulability}.

One route to bridging this divide builds on the linear combination of unitaries (LCU) primitive~\cite{childs2012}. A number of variational ans\"atze incorporate LCUs, optimising the underlying parameterised unitaries, the LCU coefficients, or both~\cite{khatri2024quixer,heredge_nonunitary_qml,coopmans2024sample,yao2025lcqnn,coyle2025density}. Recent work established that an LCU of barren plateau-free circuits is also free of barren plateaus, while increasing expressivity and potentially increasing classical simulation cost~\cite{khatri2025trainability}. While these results are general, a concrete example considered is that of fermionic Gaussian unitaries and free fermion simulation. In that case, the difference between the simulation cost using the best known classical method and the quantum gate count is subquadratic in the model parameters~\cite{dias2024classical}. This leaves open the question of whether there is a \emph{general} variational ansatz construction that, starting from unitaries that are easy to simulate classically, leads to a larger classical-quantum time-complexity difference while keeping the loss landscape from flattening exponentially.

To tackle this question, we introduce a \emph{stacked LCU} variational ansatz (S-LCU): $l$ sequentially composed LCUs, each a weighted sum of $k$ unitaries on $n$ qubits. Expanding the product over its $l$ layers, the S-LCU is a linear combination of $k^l$ unitaries, implemented with only $O(lkn^2)$ gates, before postselection overhead.
We denote the specialisation of the constituent unitaries to fermionic Gaussian unitaries by \emph{Free Fermion S-LCU} (FF-S-LCU), an ansatz whose explicit decomposition is a linear combination of $k^l$ fermionic Gaussian unitaries that is implementable using a polynomial number of quantum gates. This count upper-bounds the fermionic Gaussian rank of the system and sets the cost of direct classical simulation~\cite{dias2024classical}. Because it grows exponentially with the number of layers $l$, the model becomes increasingly costly for this classical method while remaining efficient to prepare with a polynomial-sized quantum circuit.

Our analysis rests on a diagrammatic technique for computing the first and second moments of the S-LCU expectation value, requiring only that the unitaries be drawn from a continuous finite-dimensional unitary representation of a compact group $\mathcal G$. Its central object is the Gram matrix of the commutant of each S-LCU layer. This provides a systematic means of constructing ans{\"a}tze tailored to the complexity and trainability requirements of a particular application and hardware platform. Specialising to the FF-S-LCU, we derive a variance bound of $\Omega(1/(n\,k^{3l}))$ for traceless quadratic observables. For a fixed number of layers, this rules out exponential cost concentration in $n$ and $k$. More generally, Appendix~\ref{app:general_observables} gives the variance for an arbitrary initial state and observable. The number of layers $l$ is therefore a tunable knob, yielding an arbitrarily high polynomial difference between the direct classical simulation cost and the quantum gate count, at a corresponding cost in trainability.

We begin by describing the S-LCU construction in Section~\ref{sec:model} and the notation we use. We then analyse the model’s loss landscape by characterising the first and second moments (Sections~\ref{subsec:first_moment} and~\ref{subsec:second_moment}) of the cost function when the unitaries are drawn from a representation of a compact group. In Section~\ref{sec:fermionic}, we specialise to the FF-S-LCU, analyse its trainability and classical simulation complexity, and provide asymptotic bounds on cost concentration. We close by discussing algorithmic precedents for sequential LCUs and staged strategies for training the S-LCU.

\section{S-LCU Description}\label{sec:model}

The linear combination of unitaries (LCU)~\cite{childs2012} is a construction used to prepare a weighted sum of unitary operations,
\begin{align}
    L=\sum_{i=1}^k a_i U_i,
    \qquad a_i\geq 0,\quad \sum_{i=1}^k a_i=1.
\end{align}
The S-LCU is obtained by the sequential composition of multiple independent LCUs, and prepares the state
\begin{align}
    \rho &:= A \rho_0 A^\dagger,\\
    A &:= L_l\cdots L_2 L_1,
    &L_j&:=\sum_{i=1}^k a_i^j U_i^j,
\end{align}
where $\rho_0$ is a normalised input state, $a_i^j\geq0$ with $\sum_{i=1}^k a_i^j=1$ for every $j$, and $U_i^j$ is a unitary acting on $n$ qubits. Since an LCU is not generally unitary, $\rho$ may be subnormalised, with a postselection success probability given by $p_{\mathrm{s}}:=\tr(\rho)$. For a Hermitian observable $O$, we analyse the unnormalised expectation value
\begin{align}\label{eq:cost_definition}
    m:=\tr(\rho O).
\end{align}

When deriving our results, we make extensive use of diagrammatic notation, an introduction to which can be found in Appendix~\ref{app:diagrams}. We define $V_i^j := a_i^j U_i^j$. A single LCU may then be denoted as
\begin{align}
    \input{\tikzitpathtikzfigs/slcu_single.tikz} := \sum_{i=1}^{k} V^j_i,
\end{align}
where the shaded box indicates a sum, indexed by $i$.
Using this notation, the state prepared by the S-LCU may be represented as
\begin{align}
    \rho = \input{\tikzitpathtikzfigs/slcu_borderless.tikz}.
\end{align}
Note that each layer of the S-LCU gets its own shaded box, indicating that the sum is taken independently over each layer.

\section{The S-LCU Loss Landscape}
Barren plateaus were originally defined as an exponential decay of the variance of the gradient with system size~\cite{mcclean2018barren}. It is now common practice to reason instead about the variance of the cost function itself, since exponential cost concentration is known, under the conditions of Ref.~\cite{arrasmith2022equivalence}, to accompany a barren plateau. We therefore use the variance of the expectation value $m$ as a trainability diagnostic:
\begin{align}
    \var[m] = \var[\tr(\rho O)] = E[\tr(\rho O)^2] - E[\tr(\rho O)]^2.
\end{align}
Calculating these expectations requires defining the distributions over which the moments are taken. For the unitaries, we follow the standard practice of using the Haar distribution over the relevant group~\cite{mcclean2018barren,cerezo2021cost}. The LCU coefficients are drawn independently for each layer from the uniform Dirichlet distribution. These choices are maximum entropy distributions, uniform over the parameterised spaces, and model bias-free initialisation strategies used by practitioners without further priors on the target circuit. Thus,
\begin{align}
    U_i^j &\stackrel{\mathrm{i.i.d.}}{\sim} \mathrm{Haar}(\mathcal{G}),\\
    \mathbf a^j &\stackrel{\mathrm{i.i.d.}}{\sim} \mathrm{Dir}(1,\ldots,1),
\end{align}
with the two sampled families mutually independent.

In the following sections, we derive expressions for the first and second moments of the S-LCU expectation value in terms of the commutants of the group from which the unitaries are drawn. We specialise to fermionic Gaussian unitaries in Section~\ref{sec:fermionic}.

\subsection{First Moment}\label{subsec:first_moment}
The first moment of the expectation value $m$ is
\begin{align}
    E[m] = E[\tr(\rho O)] = E\left[\tr \left( A\rho_0 A^\dagger O \right)\right].
\end{align}
Utilising the diagrammatic representation of the S-LCU from the previous section, the first moment may be restated as
\begin{align}\label{eq:slcu_tr}
    E \left[\input{\tikzitpathtikzfigs/slcu_bless_otr.tikz}\right].
\end{align}

Defining the Bell state $\ket{\Phi} = \sum_{x\in\{0,1\}^n} \ket{x}\otimes\ket{x}$, we can use the identity
  \begin{align}
      \operatorname{tr}\left(A\rho_0 A^\dagger O\right)
      =
      \langle\Phi|
      (O\otimes I)
      (A\otimes A^*)
      (\rho_0\otimes I)
      |\Phi\rangle
  \end{align}
to push linear operators around the trace. See Appendix~\ref{app:diagrams} for more details. This yields the following diagram, equivalent to Eq.~\eqref{eq:slcu_tr},
\begin{align} \label{eq:bent_1st_moment_diagram}
    E \left[\input{\tikzitpathtikzfigs/slcu_bless_vert.tikz}\right].
\end{align}
The Haar distribution is invariant under multiplication by group elements, which forces any term in which a unitary $U$ appears without a matching $U^*$ to vanish under the expectation.
\begin{restatable}{lemma}{lemvanishing}\label{lem:vanishing}
Let $U \sim \mathrm{Haar}(\mathcal{G})$ and $a \neq b$ be nonnegative integers. If the representation of the compact group $\mathcal G$ contains a scalar element $\omega I$ satisfying $\omega^{a-b}\neq1$, then
\begin{equation}
    E\!\left[U^{\otimes a} \otimes {U^*}^{\otimes b}\right] = 0.
\end{equation}
\end{restatable}
A proof is provided in Appendix~\ref{app:vanishing}. Note that the global phase of each LCU term is physically relevant, and controls interference between terms. It is possible to augment each term with an independent $U(1)$ phase group, i.e. each unitary can be sampled from the $\mathrm{Haar}(U(1) \times \mathcal{G})$ distribution instead. The balanced moments of the augmented ensemble coincide with those of $\mathrm{Haar}(\mathcal G)$, since the phases cancel, whereas all unbalanced moments vanish by the lemma above. Apart from compactness, the calculations below do not assume a particular group.

We now apply Lemma~\ref{lem:vanishing} to the unitary expansion represented by Eq.~\eqref{eq:bent_1st_moment_diagram}, yielding
\begin{align}
    E \left[\input{\tikzitpathtikzfigs/slcu_bless_vert_shared.tikz}\right],
\end{align}
where the single shaded box per layer indicates shared indexing of the sum for both matrices in each column.
The coefficients and unitaries are independent.
Factoring out the common coefficient moment from the sum therefore yields
\begin{align}
    E[a_i^2]^{l}\cdot E \left[\input{\tikzitpathtikzfigs/slcu_bless_vert_shared_U.tikz}\right].
\end{align}
The product of independent Haar-distributed unitaries is itself Haar-distributed. Indeed, for any fixed $g\in\mathcal G$,
\begin{align}
    gU_lU_{l-1}\cdots U_1
    \ \overset{\mathrm{Haar}}{=}\ U_lU_{l-1}\cdots U_1,
\end{align}
by the left-invariance of the Haar distribution ($g U_l \overset{\mathrm{Haar}}{=} U_l$). This in turn allows us to reduce the expectation under the Haar distribution of multiple layers to the expectation of a single layer as
\begin{align}
    U_lU_{l-1}\cdots U_1 \overset{\mathrm{Haar}}{=} U.
\end{align}
Each of the $k^l$ ways of choosing one unitary per layer therefore yields the same single-unitary Haar average, giving
\begin{align}
    E[a_i^2]^{l} \cdot k^l \cdot E \left[\input{\tikzitpathtikzfigs/slcu_bless_single_U.tikz}\right].
\end{align}
To evaluate this remaining Haar average, we use the \textit{vectorised moment operator}. It is the orthogonal projector onto the tensors invariant under the group action, so expressing it in terms of a basis of the commutant converts the unitary integral directly into overlaps with the initial state and observable. This is closely related to the Weingarten calculus construction discussed in Appendix~\ref{app:weingarten}. The $t$-th vectorised moment operator is given by
\begin{align}
    M^{(t)} := \underset{U \sim \mathrm{Haar}(\mathcal G)}{E} \left[U^{\otimes t} \otimes {U^*}^{\otimes t} \right] = \sum_{i=1}^{d_t} |P^t_i \rangle\rangle\langle\langle P^t_i|,
\end{align}
where $\left\{P^t_i\right\}_i$ is a Hermitian Hilbert-Schmidt-orthonormal basis for the $t$-th order commutant of $\mathcal{G}$, of dimension $d_t$~\cite{mele2023haar}. The vectorisation operation $|\cdot\rangle\rangle$ is defined in Appendix~\ref{app:diagrams}.
Applying this identity at first order, and the formula for the moment of the Dirichlet-distributed coefficients, $E[a_i^2] = \tfrac{2}{k (k+1)}$ (Appendix~\ref{app:dirichlet}), gives
\begin{align}
    E[\tr(\rho O)]
    &= \left(\frac{2}{k+1}\right)^{\!l}
    \sum_{i=1}^{d_1} \input{\tikzitpathtikzfigs/trrhopop.tikz}\nonumber\\
    &= \left(\frac{2}{k+1}\right)^{\!l}
    \sum_{i=1}^{d_1}\tr(\rho_0P_i^1)\tr(OP_i^1).
\end{align}

\begin{resultbox}[result:firstmoment]{The S-LCU first moment}
For a compact unitary group $\mathcal G$ satisfying Lemma~\ref{lem:vanishing}, with observable $O$, and initial state $\rho_0$, the first moment is
\begin{equation}
    E[\tr(\rho O)] = \left(\frac{2}{k+1}\right)^{\!l} \sum_{i=1}^{d_1} \tr\left( \rho_0 P_i^1\right) \tr\left( O P_i^1\right),
\end{equation}
a sum over an orthonormal basis for the first-order commutant $\{P_i^1\}$ of $\mathcal{G}$.
\end{resultbox}

The dependence on $\rho_0$, $O$, and the group is therefore identical to the Haar average of a single unitary from $\mathcal G$: stacking changes only the overall factor $(2/(k+1))^{l}$. The second moment is more involved, and is tackled below.

\subsection{Second Moment}\label{subsec:second_moment}
Computing the variance of the cost function requires computing its second moment, given in the S-LCU case by
\begin{align}
    E[\tr(\rho O)^2] = E\left[\tr \left( A\rho_0 A^\dagger O \right)^2\right].
\end{align}
Restating diagrammatically,
\begin{align}
    E\left[\begin{gathered} \input{\tikzitpathtikzfigs/slcu_bless_otr.tikz}\\ \input{\tikzitpathtikzfigs/slcu_bless_otr.tikz} \end{gathered}\right]
\end{align}
By pushing operators around the trace, as in the first-moment calculation, we equivalently obtain
\begin{align}
    E\left[\begin{gathered} \input{\tikzitpathtikzfigs/slcu_bless_vert.tikz}\\ \input{\tikzitpathtikzfigs/slcu_bless_vert.tikz} \end{gathered}\right].
\end{align}

This diagram consists of $l$ vertical slices, independently and identically distributed, since they correspond to different LCUs. We denote the operator formed by layer $j$ by
\begin{align}
    \label{eq:sj}
    S_j:=\sum_{i_1,i_2,i_3,i_4=1}^k
  V_{i_1}^j\otimes (V_{i_2}^j)^*
  \otimes V_{i_3}^j\otimes (V_{i_4}^j)^*.
\end{align}
Independence between layers gives
\begin{align}
    E[S_l\cdots S_1]=E[S_l]\cdots E[S_1]=E[S_j]^l.
\end{align}
Consequently,
\begin{align}\label{eq:rhorhosliceobsobs}
    E[\tr(\rho O)^2] = \input{\tikzitpathtikzfigs/slcu_rhorho_tall.tikz} E\left[\input{\tikzitpathtikzfigs/slcu_bless_slice.tikz}\right]^l \input{\tikzitpathtikzfigs/slcu_obsobs_tall.tikz}
\end{align}
The expansion \eqref{eq:sj} introduces four indices: one choice of unitary and one choice of conjugate unitary in each of the two copies of the trace. All four are drawn from the same $j$-th LCU, but the sums choose their indices independently before the Haar average imposes pairings. 
Applying Lemma~\ref{lem:vanishing} to the four factors in a slice, every distinct unitary label must occur equally often in unconjugated and conjugated form. This leaves exactly three index-identification patterns,
\begin{align}
    \scalebox{0.8}{$\input{\tikzitpathtikzfigs/paired_aaaa.tikz} \qquad \input{\tikzitpathtikzfigs/paired_aabb.tikz} \qquad \input{\tikzitpathtikzfigs/paired_abba.tikz} $},
\end{align}
where the shaded boxes enclose like unitaries and their corresponding conjugates. These represent the pairings
\begin{align}
 \pi_1 &:~U_i\otimes U_i^*\otimes U_i\otimes U_i^*,\nonumber\\
 \pi_2 &:~U_i\otimes U_i^*\otimes U_{i'}\otimes U_{i'}^*,\nonumber\\
 \pi_3 &:~U_i\otimes U_{i'}^*\otimes U_{i'}\otimes U_i^*,
 \qquad i\ne i',
\end{align}
respectively. The paired-slice picture above may be restated as below, by permuting wires appropriately,
\begin{align}
    \scalebox{0.8}{$\input{\tikzitpathtikzfigs/paired_aaaa_swapped.tikz}$ \qquad \input{\tikzitpathtikzfigs/paired_aabb.tikz} \qquad \input{\tikzitpathtikzfigs/paired_abba_swapped.tikz}}.
\end{align}

Recall that these diagrams contain not just the unitaries, but also the coefficients drawn from the uniform Dirichlet distribution. The slice average is
\begin{align}
    E[S_j] &= k\cdot E[a_i^4]M_{\pi_1}\nonumber\\
    &\quad+k(k-1)\cdot E[a_i^2a_{i'}^2]M_{\pi_2}\nonumber\\
    &\quad+k(k-1)\cdot E[a_i^2a_{i'}^2]M_{\pi_3}, \qquad i\ne i'.
\end{align}
Note that the counts here are order-sensitive. The associated moment operators can be written compactly as
\begin{align}
    M_{\pi_1} &= \sum_{i=1}^{d_2}
    |P^{\pi_1}_i \rangle\rangle\langle\langle P^{\pi_1}_i|,\\
    M_{\pi_q} &= \sum_{i,j=1}^{d_1}
    |P^{\pi_q}_{ij} \rangle\rangle\langle\langle P^{\pi_q}_{ij}|,
    \qquad q\in\{2,3\},
\end{align}
where each expression is a sum over vectorised orthonormal basis elements,
\begin{gather}
    |P^{\pi_1}_i\rangle\rangle := \input{\tikzitpathtikzfigs/paired_aaaa_onb_half.tikz} \qquad |P^{\pi_2}_{ij}\rangle\rangle := \input{\tikzitpathtikzfigs/paired_aabb_onb_half.tikz} \\
    |P^{\pi_3}_{ij}\rangle\rangle := \input{\tikzitpathtikzfigs/paired_abba_onb_half.tikz}.
\end{gather}

We concatenate these orthonormal basis elements into an ordered, overcomplete frame for the commutant of a single S-LCU layer,
\begin{align}\label{eq:ordered-frame}
    \mathbf{P} := \left[ P^{\pi_1}_i \right]_i || \left[ P^{\pi_2}_{ij} \right]_{ij} || \left[ P^{\pi_3}_{ij} \right]_{ij},
\end{align}
of $D := d_2 + 2\cdot d_1^2$ elements. Similarly, the associated coefficients are arranged in a vector $p$. The ordering need only be consistent between the frame and the coefficient vector.
The expectation from Eq.~\eqref{eq:rhorhosliceobsobs} may then be re-expressed as
\begin{align}
    E[\tr(\rho O)^2] = \input{\tikzitpathtikzfigs/slcu_rhorho.tikz} \left[\sum_{i=1}^D p_i\input{\tikzitpathtikzfigs/slcu_bless_pipi_slice.tikz}\right]^l \input{\tikzitpathtikzfigs/slcu_obsobs.tikz}.
\end{align}
We define two matrices to aid in the calculation of this expression,
\begin{align}
    V :=&~\sum_{i=1}^{D} |i\rangle\langle\langle \mathbf{P}_i| &\in \mathbb{C}^{D\times d^4} \\
    D_b :=&~\mathrm{Diag}(p) = \sum_{i=1}^{D} p_i\cdot\ket{i}\bra{i} &\in \mathbb{R}^{D\times D}.
\end{align}

Defining $|A,A\rangle\rangle := |A\rangle\rangle \otimes |A\rangle\rangle$, we have
\begin{align}
    E[\tr(\rho O)^2] =& \langle\langle O,O| (V^\dagger D_b V)^l |\rho_0,\rho_0\rangle\rangle\\
    =& \langle\langle O,O| V^\dagger (D_b V V^\dagger)^{l-1} D_b V |\rho_0,\rho_0\rangle\rangle.
\end{align}
Observe that $G:=VV^\dagger$, with entries $G_{ij}=\langle\langle\mathbf P_i|\mathbf P_j\rangle\rangle$, is the Gram matrix containing all inner products between elements of the ordered frame.
\begin{resultbox}[result:secondmoment]{The S-LCU second moment}
For a compact unitary group $\mathcal G$ satisfying Lemma~\ref{lem:vanishing}, with observable $O$, and initial state $\rho_0$, the S-LCU second moment takes the form
\begin{equation*}
    E[\tr(\rho O)^2] = \langle\langle O,O| V^\dagger (D_b G)^{l-1} D_b V |\rho_0,\rho_0\rangle\rangle ,
\end{equation*}
consisting of the following elements,
\begin{itemize}
    \item the Gram matrix $G = VV^\dagger$;
    \item the initial state projections $V|\rho_0,\rho_0\rangle\rangle$; and
    \item the observable projections $V|O,O\rangle\rangle$.
\end{itemize}
\end{resultbox}

In the following section, we consider fermionic Gaussian unitaries and derive the Gram matrix $G$ explicitly.

\section{Fermionic Gaussian Unitaries}\label{sec:fermionic}
For $n$ fermionic modes, encoded on $n$ qubits, let $c_1,\ldots,c_{2n}$ be Hermitian Majorana operators satisfying $\{c_\mu,c_\nu\}=2\delta_{\mu\nu}I$. Up to a global phase, a parity-preserving \emph{fermionic Gaussian unitary} has the form
\begin{align}
    U=e^{-iH},\qquad
    H=\frac{i}{4}\sum_{\mu,\nu=1}^{2n}h_{\mu\nu}c_\mu c_\nu,
    \qquad h\in\mathfrak{so}(2n),
\end{align}
where
\begin{align}
    \mathfrak{so}(2n):=
    \{h\in\mathbb R^{2n\times2n}:h^T=-h\}
\end{align}
is the Lie algebra of $SO(2n)$, consisting of real antisymmetric matrices. The conjugation action of $U$ is linear,
\begin{align}
    U^\dagger c_\mu U=\sum_{\nu=1}^{2n}R_{\mu\nu}c_\nu,
    \qquad R\in SO(2n).
\end{align}
The matrix $R$ records how $U$ rotates the $2n$ Majorana operators, whereas $U$ acts on the $2^n$-dimensional Fock space. The unitaries $U$ and $-U$ induce the same rotation, so the correspondence is two-to-one. The group that retains this sign is $\operatorname{Spin}(2n)$, the double cover of $SO(2n)$. Relative global phases between LCU terms affect their interference. Accordingly, the FF-S-LCU uses the Haar distribution on
\begin{align}\label{eq:phase_augmented_group}
    \mathcal G_{\mathrm{FF}}:=U(1)\times\operatorname{Spin}(2n),
\end{align}
where the $U(1)$ factor enforces Lemma~\ref{lem:vanishing}. Appendix~\ref{app:prelim} collects the remaining free fermion definitions and identities.

Fermionic Gaussian unitaries can be classically simulated in time polynomial in the number of qubits~\cite{valiant2001quantum,terhal2002classical}. Previous work has connected their polynomial trainability to their dynamical Lie algebra~\cite{matos2023characterization,diaz2023showcasing,cerezo2023simulability,fontana2024characterizing,kokcu2022fixed,wiersema2024classification}. They are a natural choice for the S-LCU: a single linear combination already increases their expressivity at a quantifiable cost in trainability~\cite{khatri2025trainability}, and stacking allows that trade-off to be tuned further. In this section we analyse concretely the variance of the loss landscape of the FF-S-LCU: an S-LCU composed of fermionic Gaussian unitaries. We begin by reviewing the first and second order commutants of the underlying group. 

\subsubsection{First-Order Commutant}
The first-order commutant of the group of fermionic Gaussian unitaries is well-studied~\cite{diaz2023showcasing,sierant2026matchgate,lastres2026geometry}.
It is spanned by the orthonormal basis
\begin{align}
    \left\{\frac{I}{\sqrt d},\frac{P}{\sqrt d}\right\}
    =\left\{ \frac{1}{\sqrt{d}}\cdot\input{\tikzitpathtikzfigs/comm1_I.tikz}, \frac{1}{\sqrt{d}}\cdot\input{\tikzitpathtikzfigs/comm1_P.tikz} \right\},
    \qquad d=2^n.
\end{align}
Here $P$ is the \textit{parity} operator, defined as
\begin{align}
    P := Z^{\otimes n} = (-i)^n \cdot c_1c_2 \dots c_{2n}.
\end{align}

\subsubsection{Second-Order Commutant}
We use the definition of the second-order commutant from ~\citet{diaz2023showcasing}, which has also been characterised in Refs.~\cite{sierant2026matchgate,lastres2026geometry}:
\begin{align}
    \left\{ \input{\tikzitpathtikzfigs/qk_0_def.tikz} \right\}_{\kappa=0}^{2n} \cup \quad \left\{ \input{\tikzitpathtikzfigs/qk_1.tikz} \right\}_{\kappa=0}^{2n},
\end{align}
where
\begin{align}
    \input{\tikzitpathtikzfigs/qk_0_def.tikz} &:= \mathcal{N}_\kappa \sum_{s \in \binom{[2n]}{\kappa}} c^s \otimes c^s,\\
    \input{\tikzitpathtikzfigs/qk_1.tikz}  &:= i^{\kappa\bmod 2} \cdot \input{\tikzitpathtikzfigs/qk_1_def.tikz},
\end{align}
where $\binom{[2n]}{\kappa}:=\{s\subseteq[2n]:|s|=\kappa\}$ and $\mathcal{N}_\kappa := \left(2^n \sqrt{\binom{2n}{\kappa}}\right)^{-1}$. Algebraically, the second line is
$Q_\kappa^1=i^{\kappa\bmod 2}(I\otimes P)Q_\kappa^0$.

\subsubsection{Gram Matrix}\label{subsec:gram-matrix}
The full Gram matrix $G$ for the FF-S-LCU, expressed in the ordered frame $\mathbf{P}$ defined in Eq.~\eqref{eq:ordered-frame}, is given below. We define the following shorthand:
\begin{gather}
    \nu_\kappa := \frac{\sqrt{\binom{2n}{\kappa}}}{d}, \quad
    \sigma_\kappa := (-1)^{\lfloor \kappa/2 \rfloor}, \\
    \epsilon_\kappa := (-1)^{\kappa(\kappa+1)/2}.        
\end{gather}
where $d = 2^n$.

Derivations for each scalar element of the matrix are provided in Appendix~\ref{app:grammatrix}.
\begin{center}
    \scalebox{0.65}{
    \begin{tabular}{ccc!{\color{gray!60}\vrule width 0.4pt}cccc!{\color{gray!60}\vrule width 0.4pt}cccc}
        &
        \scalebox{0.5}{\input{\tikzitpathtikzfigs/ordered_vects/8_qk0_m.tikz}} &
        \scalebox{0.5}{\input{\tikzitpathtikzfigs/ordered_vects/9_qk1_m.tikz}} &
        \scalebox{0.5}{\input{\tikzitpathtikzfigs/ordered_vects/0_cc_ii_m.tikz}} &
        \scalebox{0.5}{\input{\tikzitpathtikzfigs/ordered_vects/1_cc_pi_m.tikz}} &
        \scalebox{0.5}{\input{\tikzitpathtikzfigs/ordered_vects/2_cc_ip_m.tikz}} &
        \scalebox{0.5}{\input{\tikzitpathtikzfigs/ordered_vects/3_cc_pp_m.tikz}} &
        \scalebox{0.5}{\input{\tikzitpathtikzfigs/ordered_vects/4_dd_ii_m.tikz}} &
        \scalebox{0.5}{\input{\tikzitpathtikzfigs/ordered_vects/5_dd_pi_m.tikz}} &
        \scalebox{0.5}{\input{\tikzitpathtikzfigs/ordered_vects/6_dd_ip_m.tikz}} &
        \scalebox{0.5}{\input{\tikzitpathtikzfigs/ordered_vects/7_dd_pp_m.tikz}} \\[10pt]
        \scalebox{0.5}{\input{\tikzitpathtikzfigs/ordered_vects/8_qk0.tikz}} & $\delta_{\kappa,\kappa'}$ & $\cdot$ & $\cdot$ & $\cdot$ & $\cdot$ & $\cdot$ & $\cdot$ & $\cdot$ & $\cdot$ & $\cdot$ \\[10pt]
        \scalebox{0.5}{\input{\tikzitpathtikzfigs/ordered_vects/9_qk1.tikz}} & 0 & $\delta_{\kappa,\kappa'}$ & $\cdot$ & $\cdot$ & $\cdot$ & $\cdot$ & $\cdot$ & $\cdot$ & $\cdot$ & $\cdot$ \\[10pt]
        \noalign{\vskip -5pt}\arrayrulecolor{gray!60}\cline{2-11}\arrayrulecolor{black}\noalign{\vskip 5pt}
        \scalebox{0.5}{\input{\tikzitpathtikzfigs/ordered_vects/0_cc_ii.tikz}} & $\delta_{0,\kappa}$ & 0 & 1 & $\cdot$ & $\cdot$ & $\cdot$ & $\cdot$ & $\cdot$ & $\cdot$ & $\cdot$ \\[10pt]
        \scalebox{0.5}{\input{\tikzitpathtikzfigs/ordered_vects/1_cc_pi.tikz}} & 0 & $(-1)^n\delta_{2n,\kappa}$ & 0 & 1 & $\cdot$ & $\cdot$ & $\cdot$ & $\cdot$ & $\cdot$ & $\cdot$ \\[10pt]
        \scalebox{0.5}{\input{\tikzitpathtikzfigs/ordered_vects/2_cc_ip.tikz}} & 0 & $\delta_{0,\kappa}$ & 0 & 0 & 1 & $\cdot$ & $\cdot$ & $\cdot$ & $\cdot$ & $\cdot$ \\[10pt]
        \scalebox{0.5}{\input{\tikzitpathtikzfigs/ordered_vects/3_cc_pp.tikz}} & $(-1)^n\delta_{2n,\kappa}$ & 0 & 0 & 0 & 0 & 1 & $\cdot$ & $\cdot$ & $\cdot$ & $\cdot$ \\[10pt]
        \noalign{\vskip -5pt}\arrayrulecolor{gray!60}\cline{2-11}\arrayrulecolor{black}\noalign{\vskip 5pt}
        \scalebox{0.5}{\input{\tikzitpathtikzfigs/ordered_vects/4_dd_ii.tikz}} & $\sigma_\kappa \nu_\kappa$ & 0 & $\frac{1}{d}$ & 0 & 0 & $\frac{1}{d}$ & 1 & $\cdot$ & $\cdot$ & $\cdot$ \\[10pt]
        \scalebox{0.5}{\input{\tikzitpathtikzfigs/ordered_vects/5_dd_pi.tikz}} & 0 & $i^{\kappa\bmod 2}\sigma_\kappa \nu_\kappa$ & 0 & $\frac{1}{d}$ & $\frac{1}{d}$ & 0 & 0 & 1 & $\cdot$ & $\cdot$ \\[10pt]
        \scalebox{0.5}{\input{\tikzitpathtikzfigs/ordered_vects/6_dd_ip.tikz}} & 0 & $i^{\kappa\bmod 2}\epsilon_\kappa \nu_\kappa$ & 0 & $\frac{1}{d}$ & $\frac{1}{d}$ & 0 & 0 & 0 & 1 & $\cdot$ \\[10pt]
        \scalebox{0.5}{\input{\tikzitpathtikzfigs/ordered_vects/7_dd_pp.tikz}} & $\epsilon_\kappa \nu_\kappa$ & 0 & $\frac{1}{d}$ & 0 & 0 & $\frac{1}{d}$ & 0 & 0 & 0 & 1
    \end{tabular}
    }
\end{center}
We plot only the elements on or below the principal diagonal, with $G^\dagger = G$ specifying the rest.

\subsubsection{Projections of Initial State and Observable}
Result~\ref{result:secondmoment} contracts the layer operator with the initial state and observable. The initial state boundary vector needed for this contraction has entries
\begin{align}
    \frac{1}{d} \Big[
        &\tr(\rho_0)^2,~
        \tr(P\rho_0)\tr(\rho_0),~
        \tr(P\rho_0)\tr(\rho_0),\nonumber\\
        &\tr(P\rho_0)^2,~
        \tr(\rho_0^2),~
        \tr(P\rho_0^2),~
        \tr(P\rho_0^2), \nonumber\\
        &\tr(P\rho_0P\rho_0),~
        \Ptil_\kappa(\rho_0),~
        \Ctil_\kappa(\rho_0)
    \Big]
\end{align}
where
\begin{align}
    \Ptil_\kappa(\rho_0) := \tr(\rho_0^{\otimes 2} Q_{\kappa}^0), \quad
    \Ctil_\kappa(\rho_0) := \tr(\rho_0^{\otimes 2} Q_{\kappa}^1).
\end{align}
The observable boundary vector $V |O,O\rangle\rangle$, which closes the other side of the same contraction, has the same form with $\rho_0$ replaced by $O$.

\subsection{FF-S-LCU Loss Landscape}

Using the Gram matrix $G$, coefficient matrix $D_b$, and the initial state and observable vectors derived above,
the variance $\var[m]$ of the FF-S-LCU may be expressed as a function of the number of qubits $n$, LCU term count $k$, and depth $l$.

\begin{restatable}{theorem}{thmbound}\label{thm:bound}
For a traceless quadratic observable $O$ with $\tr(O^2)=d$, $\rho_0 = \ketbra{0}{0}^{\otimes n}$, $n \geq 3$,
with the variance taken over the joint distribution $\mathrm{Haar}(\mathcal G_{\mathrm{FF}})^{\otimes kl} \times \mathrm{Dir}(1,\dots,1)^{\otimes l}$,
\[
    \var[m]\geq
    \frac{1}{2n-1}
    \left(\frac{24}{(k+1)(k+2)(k+3)}\right)^{\!l}
    \in\Omega\!\left(\frac{1}{n\, k^{3l}}\right).
\]
\end{restatable}
A proof of this bound is provided in Appendix~\ref{app:general_observables} (via Theorem~\ref{thm:factorise}). The mean vanishes because $\tr(O)=\tr(PO)=0$, so the variance equals the second moment used in the proof. Appendix~\ref{app:general_observables} explains how the same methods extend to arbitrary initial states and observables.

For a fixed number of layers, the lower bound decreases polynomially in both the number of qubits and the number of terms in each LCU, ruling out exponential cost concentration in either variable. This conclusion does not apply if $l$ itself grows with the problem size. In Figure~\ref{fig:variance-by-qubit-count} we plot the exact variance against the number of qubits for different values of $l$.

\onecolumngrid
\begin{center}
    \includegraphics[width=\textwidth]{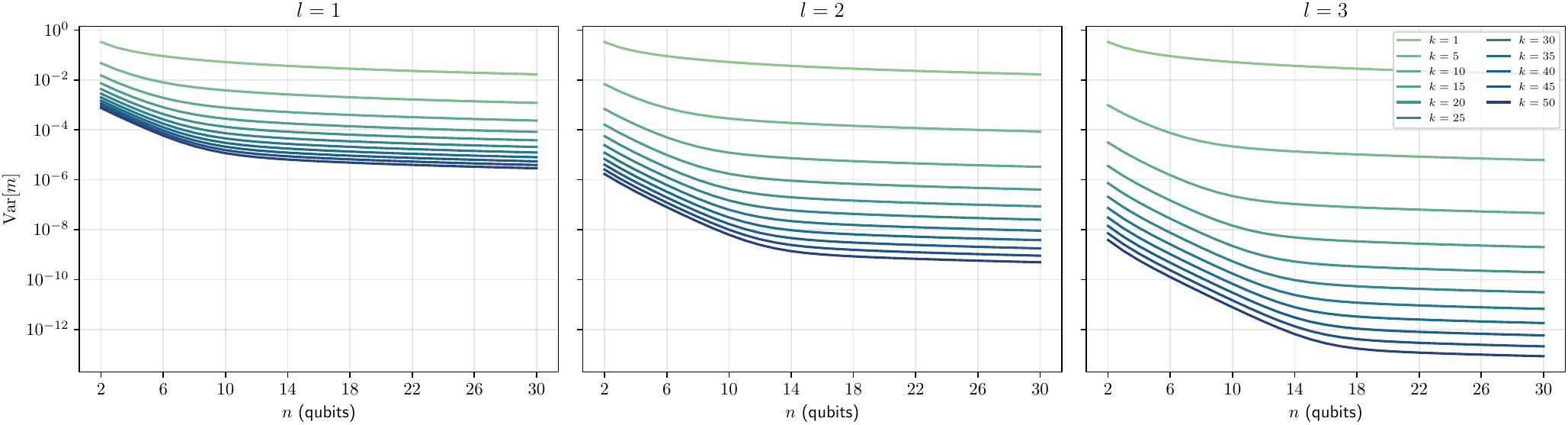}
    \captionof{figure}{Variance of the unnormalised expectation value $\var[m]$ as a function of the number of qubits $n$, for $l \in \{1,2,3\}$ and varying LCU term count $k$. Observable $O = Z_1$, initial state $\rho_0 = \ketbra{0}{0}^{\otimes n}$.}
    \label{fig:variance-by-qubit-count}
\end{center}
\twocolumngrid

\begin{center}
    \includegraphics[width=\columnwidth]{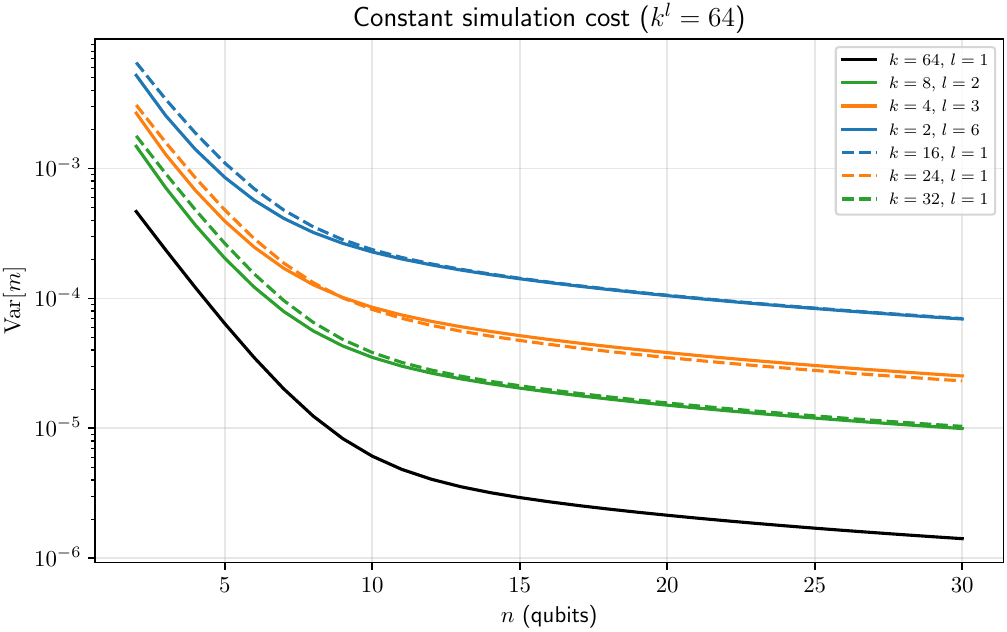}
    \captionof{figure}{Variance for FF-S-LCU configurations with constant expanded term count $k^l = 64$. Solid lines: $(k,l)$ pairs with $k^l = 64$. Dashed lines: single-layer models with comparable $k$. Distributing a fixed term count over more layers produces a larger variance in these examples.}
    \label{fig:variance-fixed-term-count}
\end{center}

Figure~\ref{fig:variance-fixed-term-count} plots the variance as a function of the number of qubits for various $(k,l)$ combinations with $k^l=64$. For the configurations shown, keeping the expanded term count (and therefore direct simulation cost) constant while increasing the number of layers reduces cost concentration.

\subsection{Time Complexity: Classical and Quantum}

Free fermion dynamics can be classically simulated in polynomial time~\cite{valiant2001quantum,terhal2002classical}. A single fermionic Gaussian unitary on $n$ qubits is fully characterised, up to phase, by an $SO(2n)$ rotation matrix, and overlaps between fermionic Gaussian states can be computed in $O(n^3)$ time, provided the relative phase is tracked.
Recent work has extended efficient classical simulation to states of bounded \emph{Gaussian rank}, where the rank is the minimum number of fermionic Gaussian states in a superposition representing the state~\cite{dias2024classical,reardonsmith2024improved,cudby2024gaussian}. Given an explicit decomposition into $r$ Gaussian states, the direct cost of computing a fixed-complexity local expectation value is $O(r^2n^3)$, since each of the $r^2$ cross-terms requires a phased overlap computation in $O(n^3)$ classical operations.

For the S-LCU model, expanding the product of $l$ layers yields a superposition of $k^l$ fermionic Gaussian states. The resulting state has fermionic Gaussian rank at most $k^l$, and direct pairwise evaluation of these overlaps yields a complexity
\begin{align}\label{eq:classical_upper_bound}
    t_{\mathrm{direct}} \in O\!\left((l+k^{2l})n^3\right)
    =O(k^{2l}n^3)\quad (k\geq2).
\end{align}
We are not aware of a method that exploits the stacked structure to lower the exponent in $k$. On the quantum side, each layer requires implementing an LCU of $k$ fermionic Gaussian unitaries, realisable in $O(kn^2)$ two-qubit gates~\cite{babbush2018encoding}. Before accounting for postselection or amplitude amplification, the total gate count for $l$ layers is therefore
\begin{align}\label{eq:quantum_gate_upper_bound}
    N_{\mathrm{gates}}\in O(lkn^2).
\end{align}

This runtime analysis does not account for the postselection cost associated with implementing non-unitary operations. Depending on the application, the runtime for obtaining successful shots may scale with the inverse success probability. We comment on this in Section~\ref{sec:discussion}. The variance calculation, however, does account for the postselection probability, since the subnormalisation of $\rho=A\rho_0A^\dagger$ is precisely this probability and $m=\tr(\rho O)$ is weighted by it. No separate correction is therefore required for the variance bound. 

\begin{resultbox}[result:tradeoff]{The FF-S-LCU tradeoff}
For a traceless quadratic observable $O$ with $\tr(O^2)=d$, $\rho_0 = \ketbra{0}{0}^{\otimes n}$, and $n\geq3$, the unnormalised cost variance satisfies
\begin{align}
    \var[m] &\in \Omega\!\left(\frac{1}{n\,k^{3l}}\right).
\end{align}
For fixed $l$, direct classical simulation using term expansion has the upper bound
\begin{align}
    t_{\mathrm{direct}} &\in O\!\left((l+k^{2l})n^3\right),
\end{align}
while the quantum circuit's gate count remains linear in $k$ before postselection overhead, scaling as
\begin{align}
    N_{\mathrm{gates}} &\in O\!\left(l k n^2\right).
\end{align}
\end{resultbox}

\section{Discussion: Practical Settings and Training Strategies}
\label{sec:discussion}

Sequentially composed LCUs already occur in several quantum algorithms. Truncated-Taylor Hamiltonian simulation divides an evolution into short-time segments and implements every segment by an LCU followed by robust oblivious amplitude amplification~\cite{berry2015taylor}. The truncated-Dyson algorithm uses the same segmented structure for time-dependent Hamiltonians~\cite{kieferova2019dyson}, while multiproduct formula simulation implements each short-time step as an LCU of product formulas and repeats the amplified step~\cite{childs2012,low2019multiproduct}. More abstractly, LCU constructions are a standard route to block encodings, and block encodings can be composed to encode matrix products~\cite{gilyen2019qsvt,chakraborty2019blockpowers}. These examples are not variational ans{\"a}tze, but they establish that a product of LCU-derived maps is a natural circuit primitive. They also illustrate an important implementation condition: the cited algorithms avoid postselecting each individual LCU. Instead, they keep the construction coherent and control the normalisation, or make each segment approximately unitary and amplify it before applying the next segment.

The layered parameterisation suggests alternatives to the fully random initialisation analysed above. One may first optimise a single LCU, append a new layer initialised near the identity, and then resume optimisation. Incrementally increasing circuit depth has been observed to retain larger gradients than training the full circuit from scratch~\cite{skolik2021layerwise}. The same idea can be applied term-wise within a layer. Candidate unitaries can be added according to a gradient or residual criterion, assigned a small nonzero coefficient, and then optimised together with the existing terms. This is analogous in spirit to adaptive variational methods that grow an ansatz one operator at a time~\cite{grimsley2019adaptive}, but here the added resource is an LCU term; the expanded term count upper-bounds the fermionic Gaussian rank. These warm-start strategies keep the effective depth and term count small when initialising. They do not, however, constitute a guarantee during optimisation: our present bounds concern independent uniform initialisation (as captured by the Haar/Dirichlet distributions), and analysing structured, correlated initialisations and the gradients encountered throughout this staged optimisation remains an open problem.

\section{Conclusion}
Our primary contribution is the S-LCU, a stacked linear combination of unitaries variational ansatz, together with a diagrammatic framework for characterising its loss landscape when the unitaries are drawn from a compact group $\mathcal G$. This gives practitioners a systematic way to tailor the complexity-trainability balance to their application and hardware.

Combining this framework with known results on the linear and quadratic commutants of fermionic Gaussian unitaries, we bound the rate of decay of the loss-landscape variance of the FF-S-LCU for traceless quadratic observables. For a fixed number of layers, the variance is bounded from below by an inverse polynomial in both the number of terms per layer and the number of qubits. We additionally derive an exact expression for arbitrary normalised initial states and Hermitian observables in Appendix~\ref{app:general_observables}.

Increasing the number of layers yields an arbitrarily high polynomial separation (in the per-layer term count $k$) between the simulation cost $O(k^{2l}n^3)$ of the best classical method known to us, and the coherent gate count $O(lkn^2)$, while the variance lower bound decreases at the corresponding rate $k^{-3l}$. Note, however, that while the postselection probability is accounted for in our variance analysis, it needs to be controlled in practical applications. Further work would be necessary to realise a practical separation from classical methods, which must contend with postselection probabilities, possible classical algorithms that exploit the stacked structure, gate execution times on real quantum devices, and sources of noise such as decoherence and imprecise readout.

\bibliography{references}

\onecolumngrid
\appendix

\section{Reading the Diagrams}\label{app:diagrams}
All diagrams in the paper are tensor networks, with minimal additional decoration to represent sums. This follows a long tradition of diagrammatic representations in quantum computing, finding use in tensor networks and categorical quantum mechanics~\cite{penrose1971applications,abramsky2004categorical,coecke2017picturing,biamonte2017tensor}.

Linear algebraic expressions can be represented using diagrams composed of wires and boxes. Wires carry finite-dimensional Hilbert spaces, while boxes represent linear maps. For example,
\begin{align}
    A: \mathbb{C}^n \rightarrow \mathbb{C}^m
    \quad\longleftrightarrow\quad \input{\tikzitpathtikzfigs/tutorial/A.tikz}.
\end{align}
When the system acted upon is unambiguous, the system size is omitted from the notation:
\begin{align}
    \input{\tikzitpathtikzfigs/tutorial/justA.tikz}.
\end{align}
Linear maps may be composed sequentially,
\begin{align}
    BA = \input{\tikzitpathtikzfigs/tutorial/AseqB.tikz},
\end{align}
or in parallel,
\begin{align}
    A \otimes B = \input{\tikzitpathtikzfigs/tutorial/AparB.tikz}.
\end{align}
The unnormalised maximally entangled vector is represented by a ``cup'':
\begin{align}
    \ket{\Phi} = \sum_{x\in\{0,1\}^n} \ket{x}\otimes\ket{x}
    = \input{\tikzitpathtikzfigs/tutorial/bell.tikz}.
\end{align}

The transpose of a linear map is represented by bending wires as
\begin{align}
    A^T = \input{\tikzitpathtikzfigs/tutorial/At.tikz}.
\end{align}
Complex conjugation and the adjoint are not represented by any change to the diagram: the box is drawn identically, with only its label decorated, as $A^*$ or $A^\dagger$.
The trace of a linear map is represented by
\begin{align}
    \tr(A) = \input{\tikzitpathtikzfigs/tutorial/trA.tikz}.
\end{align}
The vectorisation of a linear map is
\begin{align}
    |A\rangle\rangle := \left(A \otimes I \right)\ket{\mathbf{\Phi}} = \input{\tikzitpathtikzfigs/tutorial/vecA.tikz}.
\end{align}
We denote by $|A,B\rangle\rangle$ the paired vectorisation
\begin{align}
    |A,B\rangle\rangle := |A\rangle\rangle \otimes |B\rangle\rangle = \input{\tikzitpathtikzfigs/tutorial/vecAB.tikz}.
\end{align}

The Hilbert-Schmidt inner product of two linear maps is given by
\begin{align}
    \langle A, B\rangle_{HS} = \langle\langle A|B\rangle\rangle = \tr(A^\dagger B) = \input{\tikzitpathtikzfigs/tutorial/abhs.tikz}.
\end{align}

Tensor network notation does not natively offer a suitable standard for sums of linear maps, for which we introduce shaded boxes:
\begin{align}
    \sum_i A_i = \input{\tikzitpathtikzfigs/tutorial/sumA.tikz}.
\end{align}
Two separate shaded boxes in the same diagram indicate independent sums,
\begin{align}
    \sum_i A_i \otimes \sum_j B_j = \input{\tikzitpathtikzfigs/tutorial/aibi.tikz},
\end{align}
while a single box spanning multiple linear maps indicates a shared index for the sum,
\begin{align}
    \sum_i A_i \otimes B_i = \input{\tikzitpathtikzfigs/tutorial/abi.tikz}.
\end{align}

\section{Relation to the Weingarten Calculus}\label{app:weingarten}
Our method bears resemblance to the Weingarten calculus~\cite{weingarten1978asymptotic,collins2003moments,collins2006integration}, which also computes Haar averages from the Gram matrix of a group commutant~\cite{mele2023haar}. Our method differs in two respects. The Weingarten calculus builds its Gram matrix from a basis of the commutant, whereas our diagrammatic pairings give an overcomplete frame. It then inverts this Gram matrix, whereas we instead raise the Dirichlet-weighted Gram matrix $D_bG$ to a power, one factor for each layer. Since our frame is overcomplete, its Gram matrix is singular; this does not matter here because no inverse is required.

\section{Vanishing of Unbalanced Haar Moments}\label{app:vanishing}
\lemvanishing*
\begin{proof}
By left-invariance of the Haar measure, for any fixed $g \in \mathcal{G}$,
\begin{equation}
    E\!\left[U^{\otimes a} \otimes {U^*}^{\otimes b}\right] = E\!\left[(gU)^{\otimes a} \otimes {(gU)^*}^{\otimes b}\right].
\end{equation}
Taking $g = \omega I$ gives
\begin{align}
    E[U^{\otimes a} \otimes {U^*}^{\otimes b}] = \omega^{a-b} E[U^{\otimes a} \otimes {U^*}^{\otimes b}].
\end{align}
Since $\omega^{a-b} \neq 1$,
\begin{align}
    E[U^{\otimes a} \otimes {U^*}^{\otimes b}] = 0.
\end{align}
\end{proof}

\section{Free Fermion Definitions and Identities}\label{app:prelim}
In this section we introduce the definitions and identities required to prove our results involving fermionic Gaussian unitaries.

\subsection{Definitions}
A free fermion system on $n$ modes is described by creation and annihilation operators $a_j^\dagger, a_j$ ($j = 1, \dots, n$) obeying the canonical anticommutation relations (CAR),
\begin{align}
    \{a_j, a_k^\dagger\} = \delta_{jk}\,\mathbb{I}, \qquad \{a_j, a_k\} = 0,
\end{align}
where $\{A,B\} := AB + BA$ is the anticommutator. It is convenient to work instead with the $2n$ Hermitian \emph{Majorana operators}
\begin{align}
    c_{2j-1} &:= a_j + a_j^\dagger,
    & c_{2j} &:= i\,(a_j^\dagger-a_j),
\end{align}
which obey
\begin{align}
    \{ c_\mu, c_\nu\} = 2 \delta_{\mu,\nu}\,\mathbb{I}, \qquad c_\mu^2 = \mathbb{I},
\end{align}
so that each $c_\mu$ is its own inverse and distinct Majoranas anticommute. A free fermion system evolves under a quadratic Majorana Hamiltonian
\begin{align}
    H = \frac{i}{4}\sum_{\mu,\nu=1}^{2n} h_{\mu\nu}\,c_\mu c_\nu,
    \qquad h\in\mathbb R^{2n\times2n},\quad h^T=-h,
\end{align}
and the associated unitaries $e^{-iHt}$ realise $\operatorname{Spin}(2n)$ on Fock space, inducing rotations in $SO(2n)$ on the Majorana operators. Matchgate computation is historically connected to Valiant's perfect-matching identities~\cite{valiant2001quantum}.

A Majorana string $c^s$ denotes multiple sequentially composed Majorana operators, indexed by an ordered subset $s=\{\nu_1<\cdots<\nu_\kappa\}\subseteq[2n]$:
\begin{align}
    \kappa&:=|s|, & c^s&:=c_{\nu_1}c_{\nu_2}\cdots c_{\nu_\kappa}.
\end{align}
The parity operator $P$ is the full Majorana string $c_1c_2\cdots c_{2n}$, including the phase that makes it equal to $Z^{\otimes n}$:
\begin{align}
    P &:= Z^{\otimes n} = (-i)^{n} c_1c_2 \dots c_{2n}\\
    P^2 &= I.
\end{align}
\begin{align}
    \input{\tikzitpathtikzfigs/qk_0_def.tikz} &:= \mathcal{N}_\kappa \sum_{s \in \binom{[2n]}{\kappa}} c^s \otimes c^s
\end{align}
\begin{align}
    \input{\tikzitpathtikzfigs/qk_1.tikz} &:= i^{\kappa\bmod 2}\cdot\input{\tikzitpathtikzfigs/qk_1_def.tikz}
\end{align}
Equivalently, $Q_\kappa^1=i^{\kappa\bmod 2}(I\otimes P)Q_\kappa^0$.

\subsection{Identities}
We make use of the following general identities about majorana operators, the parity operator, and their interaction.
\begin{align}
    {c^s}^{\dagger}
    =(-1)^{\kappa(\kappa-1)/2}c^s
    =(-1)^{\lfloor\kappa/2\rfloor}c^s.
\end{align}
\begin{align}
    \tr(c^sc^s) = (-1)^{\lfloor \kappa / 2 \rfloor} \cdot \tr(c^s {c^s}^\dagger) = (-1)^{\lfloor \kappa / 2 \rfloor} \cdot d
\end{align}
\begin{align}
    \tr(P c^{[2n]}) = (-i)^n \tr(c^{[2n]} c^{[2n]}) = (-i)^n \cdot (-1)^{n\cdot(2n-1)} \tr(c^{[2n]} {c^{[2n]}}^\dagger) = (-i)^n \cdot (-1)^{n\cdot(2n-1)} \cdot d
\end{align}
\begin{align}
    \tr(Q_{\kappa}^0) = \delta_{0, \kappa} \cdot \mathcal{N}_\kappa \cdot \tr(I_d \otimes I_d) = \delta_{0, \kappa} \cdot\frac{1}{d} \cdot d^2 = \delta_{0, \kappa} \cdot d
\end{align}
\begin{align}
    \input{\tikzitpathtikzfigs/twist_ii_qk0_cs_tr.tikz} = \tr(c^s c^s) = (-1)^{\kappa(\kappa-1)/2} \tr(c^s {c^s}^\dagger) = (-1)^{\kappa(\kappa-1)/2} \cdot d
\end{align}
\begin{align}
    Pc^s = (-1)^{\kappa} \cdot c^s P
\end{align}

\section{Gram Matrix Elements for the Fermionic Gaussian Group}\label{app:grammatrix}

\subsection{First-Order Terms}
We evaluate gram matrix elements associated with the first-order commutant of the fermionic Gaussian group. All terms here are in $\{0,1,\frac{1}{d}\}$.

$G_{3,3}$:
\begin{align}
    \frac{1}{d^2}\cdot\input{\tikzitpathtikzfigs/cccc_iiii.tikz} = \frac{1}{d^2}\tr(I)\tr(I) = \frac{1}{d^2}\cdot d^2 = 1
\end{align}
$G_{3,4}$:
\begin{align}
    \frac{1}{d^2}\cdot\input{\tikzitpathtikzfigs/cccc_iipi.tikz} = \frac{1}{d^2}\tr(P)\tr(I) = 0
\end{align}
$G_{3,5}$:
\begin{align}
    \frac{1}{d^2}\cdot\input{\tikzitpathtikzfigs/cccc_iiip.tikz} = \frac{1}{d^2}\tr(I)\tr(P) = 0
\end{align}
$G_{3,6}$:
\begin{align}
    \frac{1}{d^2}\cdot\input{\tikzitpathtikzfigs/cccc_iipp.tikz} = \frac{1}{d^2}\tr(P)\tr(P) = 0
\end{align}

$G_{7,7}$:
\begin{align}
    \frac{1}{d^2}\cdot\input{\tikzitpathtikzfigs/dddd_iiii.tikz} = \frac{1}{d^2}\tr(I)\tr(I) = \frac{1}{d^2}\cdot d^2 = 1
\end{align}
$G_{7,8}$:
\begin{align}
    \frac{1}{d^2}\cdot\input{\tikzitpathtikzfigs/dddd_iipi.tikz} = \frac{1}{d^2}\tr(P)\tr(I) = 0
\end{align}
$G_{7,9}$:
\begin{align}
    \frac{1}{d^2}\cdot\input{\tikzitpathtikzfigs/dddd_iiip.tikz} = \frac{1}{d^2}\tr(I)\tr(P) = 0
\end{align}
$G_{7,10}$:
\begin{align}
    \frac{1}{d^2}\cdot\input{\tikzitpathtikzfigs/dddd_iipp.tikz} = \frac{1}{d^2}\tr(P)\tr(P) = 0
\end{align}

$G_{3,7}$:
\begin{align}
    \frac{1}{d^2}\cdot\input{\tikzitpathtikzfigs/ccdd_iiii.tikz} = \frac{1}{d^2}\tr(I) = \frac{1}{d^2}\cdot d = \frac{1}{d}
\end{align}
$G_{3,8}$:
\begin{align}
    \frac{1}{d^2}\cdot\input{\tikzitpathtikzfigs/ccdd_iipi.tikz} = \frac{1}{d^2}\tr(P) = 0
\end{align}
$G_{3,9}$:
\begin{align}
    \frac{1}{d^2}\cdot\input{\tikzitpathtikzfigs/ccdd_iiip.tikz} = \frac{1}{d^2}\tr(P) = 0
\end{align}
$G_{3,10}$:
\begin{align}
    \frac{1}{d^2}\cdot\input{\tikzitpathtikzfigs/ccdd_iipp.tikz} = \frac{1}{d^2}\tr(P^2) = \frac{1}{d^2}\cdot d = \frac{1}{d}
\end{align}

\subsection{Second-Order Terms}
Here we evaluate the inner products $G_{\alpha,\beta}$, for all terms involving at least one second-order commutant element. The indices follow the ordering of $\mathbf{P}$.

$G_{3,1}$:
\begin{align}
    \frac{1}{d} \cdot \input{\tikzitpathtikzfigs/cc_qk0.tikz} = \frac{1}{d}\cdot\tr(Q_{\kappa}^0) =  \delta_{0, \kappa} \cdot \frac{1}{d^2} \cdot d^2 = \delta_{0, \kappa}
\end{align}
$G_{3,2}$:
\begin{align}
    \frac{1}{d}\cdot \input{\tikzitpathtikzfigs/cc_qk1.tikz} = \frac{1}{d}\cdot\tr(Q_{\kappa}^1) = 0
\end{align}

$G_{4,1}$:
\begin{align}
    \frac{1}{d}\cdot\input{\tikzitpathtikzfigs/cpc_qk0.tikz} = \frac{1}{d}\cdot\input{\tikzitpathtikzfigs/cpc_qk0_tr.tikz} = \frac{\mathcal{N}_{\kappa}}{d} \sum_{c^s \in \binom{[2n]}{\kappa}} \input{\tikzitpathtikzfigs/cpc_qk0_cs.tikz} = 0
\end{align}

$G_{4,2}$:
\begin{align}
    \frac{1}{d}\cdot\input{\tikzitpathtikzfigs/cpc_qk1.tikz} = \frac{1}{d}\cdot\input{\tikzitpathtikzfigs/cpc_qk1_tr.tikz} = \frac{i^{\kappa\bmod 2}}{d}\cdot\input{\tikzitpathtikzfigs/cpc_qk1_pp.tikz} \\
     = \frac{i^{\kappa\bmod 2} }{d} \cdot \mathcal{N}_{\kappa} \cdot \delta_{2n, \kappa} \cdot \tr(P c^{[2n]})^2 = \delta_{2n, \kappa} \cdot (-1)^n
\end{align}

$G_{5,1}$:
\begin{align}
    \frac{1}{d}\cdot\input{\tikzitpathtikzfigs/ccp_qk0.tikz} = \frac{1}{d}\cdot\input{\tikzitpathtikzfigs/ccp_qk0_tr.tikz} = \frac{1}{d}\cdot\mathcal{N}_{\kappa} \sum_{c^s \in \binom{[2n]}{\kappa}} \input{\tikzitpathtikzfigs/ccp_qk0_cs_tr.tikz} = 0
\end{align}

$G_{5,2}$:
\begin{align}
    \frac{1}{d}\cdot\input{\tikzitpathtikzfigs/ccp_qk1.tikz} = \frac{1}{d}\cdot\input{\tikzitpathtikzfigs/ccp_qk1_tr.tikz} = \frac{i^{\kappa\bmod 2}}{d} \cdot \input{\tikzitpathtikzfigs/ccp_qk1_0tr.tikz} = \frac{i^{\kappa\bmod 2}}{d}\cdot \tr(Q_{\kappa}^0) = \delta_{0, \kappa}
\end{align}

$G_{6,1}$:
\begin{align}
    \frac{1}{d}\cdot \input{\tikzitpathtikzfigs/cpcp_qk0.tikz} = \frac{1}{d}\cdot\input{\tikzitpathtikzfigs/cpcp_qk0_tr.tikz} = \delta_{2n, \kappa} \frac{1}{d^2} \tr(P c^{[2n]})^2 = \delta_{2n, \kappa} \cdot (-1)^n
\end{align}

$G_{6,2}$:
\begin{align}
    \frac{1}{d}\cdot \input{\tikzitpathtikzfigs/cpcp_qk1.tikz} = \frac{1}{d}\cdot\input{\tikzitpathtikzfigs/cpcp_qk1_tr.tikz} = \frac{1}{d}\cdot i^{\kappa\bmod 2} \cdot \input{\tikzitpathtikzfigs/cpcp_qk1_0tr.tikz} = 0
\end{align}

$G_{7,1}$:
\begin{align}
    \frac{1}{d}\cdot\input{\tikzitpathtikzfigs/twist_ii_qk0.tikz} = \frac{1}{d}\cdot\mathcal{N}_{\kappa}\sum_{c^s \in \binom{[2n]}{\kappa}} \input{\tikzitpathtikzfigs/twist_ii_qk0_cs.tikz} = \frac{1}{d}\cdot\mathcal{N}_{\kappa}\sum_{c^s \in \binom{[2n]}{\kappa}} \input{\tikzitpathtikzfigs/twist_ii_qk0_cs_tr.tikz} = \frac{ (-1)^{\lfloor \kappa/2 \rfloor}}{d}\cdot \sqrt{\binom{2n}{\kappa}}
\end{align}

$G_{7,2}$:
\begin{align}
    \frac{1}{d}\cdot\input{\tikzitpathtikzfigs/twist_ii_qk1.tikz} = \frac{1}{d}\cdot\mathcal{N}_{\kappa} \cdot i^{\kappa\bmod 2} \sum_{c^s \in \binom{[2n]}{\kappa}} \input{\tikzitpathtikzfigs/twist_ii_qk1_cs.tikz} = \frac{1}{d}\cdot\mathcal{N}_{\kappa} \cdot i^{\kappa\bmod 2} \sum_{c^s \in \binom{[2n]}{\kappa}} \input{\tikzitpathtikzfigs/twist_ii_qk1_cs_tr.tikz} = 0
\end{align}

$G_{8,1}$:
\begin{align}
    \frac{1}{d}\cdot\input{\tikzitpathtikzfigs/twist_pi_qk0.tikz} = \frac{1}{d}\cdot\mathcal{N}_{\kappa} \sum_{c^s \in \binom{[2n]}{\kappa}} \input{\tikzitpathtikzfigs/twist_ii_qk1_cs_tr.tikz} = 0
\end{align}

$G_{8,2}$:
\begin{align}
    \frac{1}{d}\cdot\input{\tikzitpathtikzfigs/twist_pi_qk1.tikz}
    &= \frac{i^{\kappa\bmod 2}}{d}\cdot
       \input{\tikzitpathtikzfigs/twist_pi_qk1_0.tikz}\nonumber\\
    &= \frac{i^{\kappa\bmod 2}}{d}\mathcal{N}_{\kappa}
       \sum_{c^s \in \binom{[2n]}{\kappa}}
       \input{\tikzitpathtikzfigs/twist_ii_qk0_cs.tikz}\nonumber\\
    &= \frac{i^{\kappa\bmod 2}(-1)^{\lfloor \kappa/2 \rfloor}}{d}
       \sqrt{\binom{2n}{\kappa}}.
\end{align}

$G_{9,1}$:
\begin{align}
    \frac{1}{d} \cdot \input{\tikzitpathtikzfigs/twist_ip_qk0.tikz} = \frac{1}{d} \cdot \mathcal{N}_{\kappa} \sum_{c^s \in \binom{[2n]}{\kappa}} \input{\tikzitpathtikzfigs/twist_ii_qk1_cs_tr.tikz} = 0
\end{align}

$G_{9,2}$:
\begin{align}
    \frac{1}{d} \cdot\input{\tikzitpathtikzfigs/twist_ip_qk1.tikz} = \frac{i^{\kappa\bmod 2}}{d} \cdot \input{\tikzitpathtikzfigs/twist_ip_qk1_0.tikz} = \frac{i^{\kappa\bmod 2}}{d} \mathcal{N}_{\kappa} \sum_{c^s \in \binom{[2n]}{\kappa}} \input{\tikzitpathtikzfigs/tr_pcpc.tikz} \nonumber\\
    = (-1)^{\kappa\bmod 2} \cdot \frac{i^{\kappa\bmod 2}}{d} \cdot \mathcal{N}_{\kappa} \sum_{c^s \in \binom{[2n]}{\kappa}} \input{\tikzitpathtikzfigs/twist_ii_qk0_cs_tr.tikz} = (-1)^{\kappa (\kappa + 1)/ 2} \cdot \frac{i^{\kappa\bmod 2}}{d}  \sqrt{\binom{2n}{\kappa}}
\end{align}

$G_{10,1}$:
\begin{align}
    \frac{1}{d} \cdot \input{\tikzitpathtikzfigs/twist_pp_qk0.tikz} = \frac{1}{d} \cdot\mathcal{N}_{\kappa} \sum_{c^s \in \binom{[2n]}{\kappa}} \input{\tikzitpathtikzfigs/tr_pcpc.tikz}
    = \frac{(-1)^{\kappa}}{d} \cdot \mathcal{N}_{\kappa} \sum_{c^s \in \binom{[2n]}{\kappa}} \input{\tikzitpathtikzfigs/twist_ii_qk0_cs_tr.tikz} \nonumber\\
    = (-1)^{\kappa(\kappa + 1)/2} \cdot \frac{1}{d} \sqrt{\binom{2n}{\kappa}}
\end{align}

$G_{10,2}$:
\begin{align}
    \frac{1}{d} \cdot \input{\tikzitpathtikzfigs/twist_pp_qk1.tikz} = \frac{i^{\kappa\bmod 2}}{d} \cdot \input{\tikzitpathtikzfigs/twist_pp_qk1_0.tikz} = \frac{i^{\kappa\bmod 2}}{d} \input{\tikzitpathtikzfigs/twist_ip_qk0.tikz} = 0
\end{align}

$G_{1,1}$:
\begin{align}
    \input{\tikzitpathtikzfigs/second_order_00.tikz} = \tr(Q_{\kappa'}^0Q_{\kappa}^0) = \delta_{\kappa, \kappa'}
\end{align}

$G_{1,2}$:
\begin{align}
    \input{\tikzitpathtikzfigs/second_order_01.tikz} = \tr(Q_{\kappa'}^1Q_{\kappa}^0) = 0
\end{align}

$G_{2,2}$:
\begin{align}
    \input{\tikzitpathtikzfigs/second_order_11.tikz} = \tr(Q_{\kappa'}^1Q_{\kappa}^1) = \delta_{\kappa, \kappa'}
\end{align}

\section{Bounds for the Second Moment of the FF-S-LCU}\label{app:general_observables}
We derive the second moment of the FF-S-LCU for arbitrary $\rho_0,O$. From Sec.~\ref{subsec:second_moment}, the single-layer average $\Lambda := E[S_1] = \sum_i(D_b)_i\dk{\vv P_i}\db{\vv P_i}$ gives
\begin{align}\label{eq:g:setup}
    E[m_O^2]=\db{O,O}\Lambda^{\,l}\dk{\rho_0,\rho_0}.
\end{align}
In terms of the moment operators $M_{\pi_1},M_{\pi_2},M_{\pi_3}$ of Sec.~\ref{subsec:second_moment},
\begin{align}\label{eq:g:split}
    \Lambda=k\cdot E[a_i^4]\cdot M_{\pi_1}+k(k-1)\cdot E[a_i^2a_j^2]\cdot\big(M_{\pi_2}+M_{\pi_3}\big).
\end{align}
Each $M_{\pi_m}$ is an orthogonal projector, and the $\pi_2,\pi_3$ ranges are spanned by that of $\pi_1$. This may be observed by inspection of the Gram matrix in Sec.~\ref{subsec:gram-matrix} and Appendix~\ref{app:grammatrix}, and provides the following properties of the projectors,
\begin{align}\label{eq:g:rules}
    M_{\pi_m}^2=M_{\pi_m}\ \ (m=1,2,3),
    \qquad
    M_{\pi_1}M_{\pi_2}=M_{\pi_2},\quad M_{\pi_1}M_{\pi_3}=M_{\pi_3},
    \qquad
    \|M_{\pi_2}M_{\pi_3}\|=2/d .
\end{align}
We define the projector composed of the first-order commutants as $R := M_{\pi_2} + M_{\pi_3}$, which commutes with the second-order commutant projector as $M_{\pi_1}R = R M_{\pi_1}$.
Expanding $\Lambda^{l}$, these properties collapse the expansion to terms composed of $M_{\pi_1}$, $R$, and a term $\Delta$, exponentially small in the number of qubits.
\begin{align}\label{eq:g:power}
    \Lambda^{l}=(k\cdot E[a_i^4])^{l}\cdot M_{\pi_1} + c_l \cdot R + \Delta,
    \qquad
    \|\Delta\|\leq 4(l-1) \cdot c_l \frac{1}{d}\Big(1+\frac4d\Big)^{l-2} ,
\end{align}
where $c_l = \big(k\cdot E[a_i^4] + k(k-1) \cdot E[a_i^2a_j^2]\big)^{l} - (k\cdot E[a_i^4])^{l}$, and
\begin{align}\label{eq:g:delta-forms}
    \Delta =\sum_{r=2}^{l}\binom{l}{r}(k\,E[a_i^4])^{l-r}(k(k-1)E[a_i^2a_j^2])^{r}\big(R^{\,r}-R\big).
\end{align}
For the $\|\Delta\|$ bound, $R$ is a projector up to exponentially small error $O(1/d)$. $R^2=R+W$ with
$W:=M_{\pi_2}M_{\pi_3}+M_{\pi_3}M_{\pi_2}$, $\|W\|\leq2\|M_{\pi_2}M_{\pi_3}\|=4/d$ (adjoints have equal norm). Self-adjointness gives
$\|R\|^2=\|R^2\|\leq\|R\|+\|W\|$, so $\|R\|\leq1+4/d$. We may decompose the matrix term appearing in the expansion of $\Delta$ as
\begin{align}
    R^{\,r}-R=(I+R+\cdots+R^{r-2})W,
\end{align}
the operator norm of which is bounded by
\begin{align}
    \|R^{\,r}-R\|\leq\tfrac{4(r-1)}{d}(1+\tfrac4d)^{r-2}.
\end{align}
Using this bound in the sum in \eqref{eq:g:delta-forms} yields the bound of \eqref{eq:g:power}.

The remaining terms may be contracted with the boundary vectors, beginning with the second-order commutant projector,
\begin{align}\label{eq:g:c1}
    \db{O,O}M_{\pi_1}\dk{\rho_0,\rho_0}
      &=\sum_{j,\kappa}\tr\!\big(O^{\otimes2}Q^{\,j}_\kappa\big)\,
        \tr\!\big(\rho_0^{\otimes2}Q^{\,j}_\kappa\big)
      =\sum_{\kappa=0}^{2n}\Big[\Ptil_\kappa(O)\,\Ptil_\kappa(\rho_0)
        +\Ctil_\kappa(O)\,\Ctil_\kappa(\rho_0)\Big].
\end{align}
Here $\Ptil_\kappa=\tfrac{\sigma_\kappa}{\sqrt{D_\kappa}}\mathcal P_\kappa$ and
$\Ctil_\kappa=\tfrac{\sigma_\kappa}{\sqrt{D_\kappa}}\mathcal C_\kappa$ are scaled and signed versions of the $\kappa$-purity
and $\kappa$-coherence of Ref.~\cite{diaz2023showcasing}, with
$\sigma_\kappa=(-1)^{\lfloor\kappa/2\rfloor}$ and $D_\kappa=\binom{2n}{\kappa}$. The terms involving the first-order moment projectors may similarly be calculated as
\begin{align}\label{eq:g:c2}
    \db{O,O}\,M_{\pi_2}\,\dk{\rho_0,\rho_0}
      &=\frac{1}{d^{2}}\Big(\tr(O)\tr(\rho_0)+\tr(PO)\tr(P\rho_0)\Big)^{\!2},\\
    \db{O,O}\,M_{\pi_3}\,\dk{\rho_0,\rho_0}
      &=\frac{1}{d^{2}}\Big(\tr(O^2)\tr(\rho_0^2)+2\tr(PO^2)\tr(P\rho_0^2)
        +\tr(POPO)\tr(P\rho_0P\rho_0)\Big).
\end{align}
Combining these gives the second moment. The error term for the second moment, $\db{O,O}\Delta\dk{\rho_0,\rho_0}$, is at most $\|\Delta\|\,\tr(O^2)\tr(\rho_0^2)$ (using $\|\dk{X,X}\| = \tr(X^2)$); with
$\tr(\rho_0^2)\leq1$ and $c_l\leq1$, and $l \ll d$, this term is of order $O(l\,2^{-n}\tr(O^2))$.

\begin{center}\begin{minipage}{0.70\textwidth}
\begin{resultbox}[result:generalsecondmoment]{The FF-S-LCU second moment}
For any normalised initial state $\rho_0$ and any Hermitian observable $O$, the second
moment is
\begin{equation*}
\begin{aligned}
    E[m_O^2]\;=\;&(k\,E[a_i^4])^{l}\sum_{\kappa=0}^{2n}
      \Big[\Ptil_\kappa(O)\,\Ptil_\kappa(\rho_0)
        +\Ctil_\kappa(O)\,\Ctil_\kappa(\rho_0)\Big]\\[2pt]
    +\;&\frac{c_l}{d^{2}}\Big[\big(\tr(O)\tr(\rho_0)+\tr(PO)\tr(P\rho_0)\big)^{2}
      +\tr(O^2)\tr(\rho_0^2)\\
    &\hphantom{{}+\frac{c_l}{d^{2}}\Big[}
      +2\tr(PO^2)\tr(P\rho_0^2)+\tr(POPO)\tr(P\rho_0P\rho_0)\Big]\\[2pt]
    &+\;O\big(l\,2^{-n}\tr(O^2)\big),
\end{aligned}
\end{equation*}
where $c_l=(k\,E[a_i^4]+k(k-1)E[a_i^2a_j^2])^{l}-(k\,E[a_i^4])^{l}$.
\end{resultbox}
\end{minipage}\end{center}

\subsection{Concrete bounds on the second moment}
The result provided above is general for the FF-S-LCU, and makes very few assumptions beyond the norm of $\rho_0$. With minimal additional restrictions, we may provide a concrete lower bound for the variance of the expectation value. Let $\vv o_R,\vv s_R,G_R$ be the restrictions of $\vv o,\vv s,G$ to the eight first-order frame elements (the $\pi_2,\pi_3$ block).

\begin{theorem}\label{thm:factorise}
Let $\tr(O)=\tr(PO)=0$ ($O$ orthogonal to the first-order commutant, so that
$E[m]=0$ and $\Var[m]=E[m^2]$) and $\vv o_R,\vv s_R\geq0$. Then
\begin{align}\label{eq:g:factorise}
    \Var[m]\ \geq\ (k\,E[a_i^4])^{l}\,\db{O,O}M_{\pi_1}\dk{\rho_0,\rho_0}
    \ =\ (k\,E[a_i^4])^{l}\,\Var[m_U],
\end{align}
where $m_U := \tr \left( U\rho_0 U^\dagger O\right)$ is the expectation value for a single Fermionic Gaussian unitary.
\end{theorem}

\begin{proof}
By \eqref{eq:g:rules}, $\Lambda^{l}=(k\,E[a_i^4])^{l}M_{\pi_1}+\sum_{r=1}^{l}\binom{l}{r}(k\,E[a_i^4])^{l-r}(k(k-1)E[a_i^2a_j^2])^{r}R^{\,r}$. Considering only the first-order term,
\begin{align}\label{eq:g:z1walk}
    \db{O,O}R^{\,r}\dk{\rho_0,\rho_0}=\vv o_R^\top\,G_R^{\,r-1}\,\vv s_R .
\end{align}
Because each $R^{\,r}$ acts only within the first-order block, this walk involves only
$\vv o_R,\vv s_R$ and $G_R$; the second-order overlaps enter solely through the retained
$M_{\pi_1}$ term. The entries of $G_R$ are traces of products of $I$ and $P$
(Appendix~\ref{app:grammatrix}), hence lie in $\{0,\tfrac1d,1\}$, and $\vv o_R,\vv s_R\geq0$
by hypothesis; so every $r\geq1$ term is non-negative and drops, giving
\eqref{eq:g:factorise}.
\end{proof}

The factors decouple: $(k\,E[a_i^4])^{l}\sim k^{-3l}$ carries the stacking, $\Var[m_U]$
the state and observable (closed form for all $\rho_0,O$~\cite{diaz2023showcasing}); any
plateau is inherited from the twirl, not created by stacking. The hypotheses hold for
even-parity $\rho_0$ and parity-preserving $O$ with $\tr(PO^2)\geq0$.

\thmbound*

\begin{proof}
A traceless quadratic $O$ is a Majorana bilinear: it commutes with $P$ and carries all
weight at degree $\kappa=2$. So $\tr(O)=\tr(PO)=0$ (the $cc$ entries of $\vv o_R$ vanish),
and via $POP=O$,
\begin{align}\label{eq:g:kappasum}
    \tr(POPO)=\tr(O^2)=d,
    \qquad
    \tr(PO^2)=0\quad(n\geq3),
\end{align}
the latter since $O^2$ has degree $\leq4<2n$. Hence $\vv o_R=(0,0,0,0,1,0,0,1)$ and
$\vv s_R=\tfrac1d(1,\dots,1)$ ($\tr(\rho_0^m)=1$), both $\geq0$, so
Theorem~\ref{thm:factorise} applies. Only $\kappa=2$ survives in $\Var[m_U]$, which
for $\tr(O^2)=d$ equals $1/(2n-1)$~\cite{diaz2023showcasing}. Substituting into
\eqref{eq:g:factorise} gives $\Var[m]\geq(k\,E[a_i^4])^{l}/(2n-1)$; with
$k\,E[a_i^4]=24/((k{+}1)(k{+}2)(k{+}3))$ (Appendix~\ref{app:dirichlet}) this is the stated
bound. The case $O=Z_1$ holds also at $n=2$, as $Z_1^2=I$ keeps $\tr(PZ_1^2)=0$.
\end{proof}

\section{Dirichlet Coefficients}\label{app:dirichlet}
For $(a_1,\dots,a_k)\sim\mathrm{Dir}(1,\dots,1)$, the uniform Dirichlet distribution over the $(k-1)$-simplex, the following moments hold for $i\neq j$:
\begin{align}
    E[a_i] &= \frac{1}{k},\\
    E[a_i^2] &= \frac{2}{k (k+1)},\\
    E[a_i^4] &= \frac{24}{k(k+1)(k+2)(k+3)},\\
    E[a_ia_j] &= \frac{1}{k(k+1)},\\
    E[a_i^2a_j^2] &= \frac{4}{k(k+1)(k+2)(k+3)}.
\end{align}

\end{document}

%% file: tikzinfo.tex
\definecolor{regboxc}{HTML}{FFFFFF}
\definecolor{tallboxc}{HTML}{FFFFFF}

\definecolor{parityc}{HTML}{FFDC85}
\definecolor{majoranac}{HTML}{94D1BE}
\definecolor{densitymatrixc}{HTML}{b1d3ef}
\definecolor{obsc}{HTML}{E09891}

\definecolor{shadec}{HTML}{ffc4ec}

\tikzstyle{tall box}=[fill=tallboxc, draw=black, shape=rectangle, tikzit draw=black, tikzit shape=rectangle, tikzit fill=white, minimum height=1cm, minimum width=0.5cm, inner sep=0mm]
\tikzstyle{regular box}=[fill=regboxc, draw=black, shape=rectangle, tikzit draw=black, tikzit shape=rectangle, tikzit fill=white, minimum height=0.5cm, minimum width=0.5cm, inner sep=0mm, text width=0.5cm, font=\small, align=center]
\tikzstyle{parity}=[fill=parityc, draw=black, shape=circle, label={center:\small{$P$}}, minimum width=4mm]
\tikzstyle{majorana}=[fill=majoranac, draw=black, shape=circle, minimum width=4mm, inner sep=0mm, font={\small}]
\tikzstyle{densitymatrix}=[fill=densitymatrixc, draw=black, shape=rectangle, tikzit draw=black, tikzit shape=rectangle, minimum height=0.5cm, minimum width=0.5cm, inner sep=0mm]
\tikzstyle{obs}=[fill=obsc, draw=black, shape=circle, label={center:\small{$O$}}, minimum width=4mm]

\tikzstyle{gate}=[shape=rectangle, text height=1.5ex, text depth=0.25ex, yshift=0.5mm, fill=white, draw=black, minimum height=5mm, yshift=-0.5mm, minimum width=5mm, font={\small}]
\tikzstyle{big gate}=[shape=rectangle, text height=1.5ex, text depth=0.25ex, yshift=0.5mm, fill=white, draw=black, minimum height=10mm, yshift=-0.5mm, minimum width=5mm, font={\small}]
\tikzstyle{Z dot}=[inner sep=0mm, minimum size=2mm, shape=circle, draw=black, fill={rgb,255: red,221; green,255; blue,221}]
\tikzstyle{Z phase dot}=[minimum size=5mm, font={\footnotesize\boldmath}, shape=rectangle, rounded corners=2mm, inner sep=0.4mm, outer sep=-2mm, scale=0.8,  draw=black, fill={rgb,255: red,221; green,255; blue,221}]
\tikzstyle{X dot}=[Z dot, shape=circle, draw=black, fill={rgb,255: red,255; green,136; blue,136}]
\tikzstyle{X phase dot}=[Z phase dot, fill={rgb,255: red,255; green,136; blue,136}, font={\footnotesize\boldmath}]
\tikzstyle{hadamard}=[fill=yellow, draw=black, shape=rectangle, inner sep=0.6mm, minimum height=1.5mm, minimum width=1.5mm]
\tikzstyle{paulibox}=[fill={rgb,255: red,221; green,221; blue,255}, draw=black, shape=rectangle, inner sep=0.6mm, minimum height=5mm, minimum width=5mm, font={\footnotesize}, text height=1.5ex, text depth=0.25ex]
\tikzstyle{vertex}=[inner sep=0mm, minimum size=1mm, shape=circle, draw=black, fill=black]
\tikzstyle{vertex set}=[inner sep=0mm, minimum size=1mm, shape=circle, draw=black, fill=white, font={\footnotesize\boldmath}]
\tikzstyle{small black dot}=[fill=black, draw=black, shape=circle, inner sep=0pt, minimum width=1.2mm]
\tikzstyle{cnot ctrl}=[fill=black, draw=black, shape=circle, inner sep=0pt, minimum width=1.2mm]
\tikzstyle{cnot targ}=[fill=white, draw=white, shape=circle, label={center:$\oplus$}, inner sep=0pt, minimum width=2.1mm]
\tikzstyle{ket}=[fill=white, draw=black, shape=regular polygon, regular polygon sides=3, regular polygon rotate=-30, scale=0.7, inner sep=1pt]
\tikzstyle{bra}=[fill=white, draw=black, shape=regular polygon, regular polygon sides=3, regular polygon rotate=30, scale=0.7, inner sep=1pt]
\tikzstyle{scalar}=[shape=rectangle, text height=1.5ex, text depth=0.25ex, yshift=0.5mm, fill=white, draw=black, minimum height=5mm, yshift=-0.5mm, minimum width=5mm, font={\small}]
\tikzstyle{clabel}=[fill=white, draw=none, shape=rectangle, inner sep=1pt]
\tikzstyle{empty diagram}=[draw={gray!40!white}, dashed, shape=rectangle, minimum width=1cm, minimum height=1cm]
\tikzstyle{amap}=[fill=white, draw=black, shape=NEbox]
\tikzstyle{amap conj}=[fill=white, draw=black, shape=NWbox]
\tikzstyle{amap adj}=[fill=white, draw=black, shape=SEbox]
\tikzstyle{amap trans}=[fill=white, draw=black, shape=SWbox]
\tikzstyle{astate}=[fill=white, draw=black,line width=1.6pt, shape=NEtriangle]
\tikzstyle{astate conj}=[fill=white, draw=black, line width=1.6pt, shape=NWtriangle]
\tikzstyle{astate adj}=[fill=white, draw=black, line width=1.6pt, shape=SEtriangle]
\tikzstyle{astate trans}=[fill=white, draw=black, line width=1.6pt, shape=SWtriangle]

\tikzstyle{hadamard edge}=[-, dashed, dash pattern=on 2pt off 0.5pt, thick, draw={rgb,255: red,68; green,136; blue,255}]
\tikzstyle{box edge}=[-, dashed, dash pattern=on 2pt off 0.5pt, thick, draw={rgb,255: red,203; green,192; blue,225}]
\tikzstyle{gray dashed edge}=[-, dashed, dash pattern=on 2pt off 0.5pt, thick, {gray!60!white}]
\tikzstyle{brace edge}=[-,  decorate, decoration={brace,amplitude=1mm,raise=-1mm}]
\tikzstyle{diredge}=[->]
\tikzstyle{double edge}=[-, double, shorten <=-1mm, shorten >=-1mm, double distance=2pt]
\tikzstyle{gray edge}=[-, {gray!60!white}]
\tikzstyle{blue edge}=[-, {blue!60!white}]
\tikzstyle{pointer edge}=[->, very thick, gray]
\tikzstyle{boldedge}=[-, line width=1.6pt, shorten <=-0.17mm, shorten >=-0.17mm]
\tikzstyle{bidir edge}=[<->, very thick, draw={rgb,255: red,191; green,191; blue,191}]
\tikzstyle{new edge style 0}=[-, fill=white]
\tikzstyle{paired boxes}=[-, draw=none, fill={rgb,255: red,164; green,219; blue,255}]

\pgfdeclarelayer{edgelayer}
\pgfdeclarelayer{nodelayer}
\pgfsetlayers{edgelayer,nodelayer,main}
\tikzstyle{none}=[inner sep=0pt]

%% file: tikzfigs/comm1_I.tikz
\begin{tikzpicture}
	\begin{pgfonlayer}{nodelayer}
		\node [style=none] (0) at (0, 0) {};
		\node [style=none] (1) at (-1, 0) {};
		\node [style=none] (2) at (1, 0) {};
	\end{pgfonlayer}
	\begin{pgfonlayer}{edgelayer}
		\draw (1.center) to (0.center);
		\draw (0.center) to (2.center);
	\end{pgfonlayer}
\end{tikzpicture}

%% file: tikzfigs/comm1_P.tikz
\begin{tikzpicture}
	\begin{pgfonlayer}{nodelayer}
		\node [style=parity] (0) at (0, 0) {};
		\node [style=none] (1) at (-1, 0) {};
		\node [style=none] (2) at (1, 0) {};
	\end{pgfonlayer}
	\begin{pgfonlayer}{edgelayer}
		\draw (1.center) to (0);
		\draw (0) to (2.center);
	\end{pgfonlayer}
\end{tikzpicture}

%% file: tikzfigs/tutorial/justA.tikz
\begin{tikzpicture}
	\begin{pgfonlayer}{nodelayer}
		\node [style=regular box] (0) at (0, 0) {$A$};
		\node [style=none] (1) at (-2, 0) {};
		\node [style=none] (2) at (2, 0) {};
	\end{pgfonlayer}
	\begin{pgfonlayer}{edgelayer}
		\draw (1.center) to (0);
		\draw (0) to (2.center);
	\end{pgfonlayer}
\end{tikzpicture}

%% file: tikzfigs/tutorial/AseqB.tikz
\begin{tikzpicture}
	\begin{pgfonlayer}{nodelayer}
		\node [style=regular box] (0) at (-0.75, 0) {$A$};
		\node [style=none] (1) at (-2, 0) {};
		\node [style=none] (2) at (2, 0) {};
		\node [style=regular box] (3) at (0.75, 0) {$B$};
	\end{pgfonlayer}
	\begin{pgfonlayer}{edgelayer}
		\draw (1.center) to (0);
		\draw (0) to (2.center);
	\end{pgfonlayer}
\end{tikzpicture}

%% file: tikzfigs/dddd_iiii.tikz
\begin{tikzpicture}
	\begin{pgfonlayer}{nodelayer}
	\end{pgfonlayer}
	\begin{pgfonlayer}{edgelayer}
		\draw (0,0.5)
			to [in=180, out=180, looseness=1.25] (0,-0.5)
			-- (3,-0.5)
			to [in=0, out=0, looseness=1.25] (3,0.5)
			-- cycle;
		\draw (0,1.5)
			to [in=180, out=180] (0,-1.5)
			-- (3,-1.5)
			to [in=0, out=0] (3,1.5)
			-- cycle;
	\end{pgfonlayer}
\end{tikzpicture}